\documentclass[12pt,letterpaper]{article}

\usepackage{amsmath, amsfonts, amssymb, amsthm}
\usepackage{xspace}
\usepackage{complexity}
\usepackage{lmodern}
\usepackage{algorithm}
\usepackage{graphicx}
\usepackage{algpseudocode}
\usepackage{url}
\usepackage{endnotes}
\usepackage{appendix}
\usepackage[footnotesize]{caption}

\newtheorem{thm}{Theorem}
\newtheorem{lem}[thm]{Lemma}
\newtheorem{cor}[thm]{Corollary}
\newtheorem{problem}{Problem}

\newcommand{\bigO}[1]{\mathcal{O}\left(#1\right)\xspace}
\newcommand{\SAG}{\mathsf{SAG}\xspace}
\DeclareMathOperator{\OG}{\mathsf{OG}}
\DeclareMathOperator{\COG}{\mathsf{COG}}

\DeclareMathOperator{\rlz}{\textrm{R}}
\DeclareMathOperator{\ncrlz}{\textrm{NR}}

\DeclareMathOperator{\sn}{\textrm{sn}}

\newcommand{\CAMSA}{\textsf{CAMSA}\xspace}

\newcommand{\Scafs}{\mathbb{S}}
\newcommand{\Ass}{\mathbb{A}}
\newcommand{\Orp}{\mathbb{O}}

\DeclareMathOperator{\AV}{\textrm{AV}}
\DeclareMathOperator{\BV}{\textrm{BV}}
\DeclareMathOperator{\NBPB}{\textrm{PB}}
\DeclareMathOperator{\CC}{\textrm{CC}}
\DeclareMathOperator{\Con}{\textrm{C}}
\DeclareMathOperator{\ST}{\textrm{ST}}

\newcommand{\rd}[1]{\protect\overrightarrow{#1}}
\newcommand{\ld}[1]{\protect\overleftarrow{#1}}
\newcommand{\ar}[1]{#1'}

\newcommand{\MAX}{\mathrm{MAX}}
\newcommand{\DNF}{\mathrm{DNF}}
\newcommand{\CUT}{\mathrm{CUT}}

\usepackage{hyperref}

\title{Orienting Ordered Scaffolds: Complexity and Algorithms}

\author{
Sergey~Aganezov\textsuperscript{1,}\footnote{These authors contributed equally},
Pavel~Avdeyev\textsuperscript{2,}\footnotemark[1],
Nikita~Alexeev\textsuperscript{3},\\
Yongwu~Rong\textsuperscript{4},
and~
Max A. Alekseyev\textsuperscript{2,5,}\footnote{Corresponding author: maxal@gwu.edu} 
\bigskip
\\
\footnotesize{\textsuperscript{1}  Department of Computer Science, John Hopkins University,}\\\footnotesize{Baltimore, MD, USA} \\
\footnotesize{\textsuperscript{2}  Department of Mathematics, The George Washington University,}\\\footnotesize{Washington, DC, USA} \\
\footnotesize{\textsuperscript{3}  International Laboratory ``Computer technologies'',}\\\footnotesize{ITMO University, Saint Petersburg, Russia} \\
\footnotesize{\textsuperscript{4}  Department of Mathematics, CUNY Queens College, Queens, NY, USA} \\
\footnotesize{\textsuperscript{5}  Department of Biostatistics and Bioinformatics,}\\\footnotesize{The George Washington University, Washington, DC, USA}
\\
}

\date{}

\begin{document}

\maketitle



 \begin{abstract}
Despite the recent progress in genome sequencing and assembly, many of the currently available assembled genomes come in a draft form. Such draft genomes consist of a large number of genomic fragments (\emph{scaffolds}), whose order and/or orientation (i.e., strand) in the genome are unknown. There exist various scaffold assembly methods, which attempt to determine the order and orientation of scaffolds along the genome chromosomes. Some of these methods (e.g., based on FISH physical mapping, chromatin conformation capture, etc.) can infer the order of scaffolds, but not necessarily their orientation. This leads to a special case of the \emph{scaffold orientation problem}  (i.e., deducing the orientation of each scaffold) with a known order of the scaffolds. 
     
We address the problem of orientating ordered scaffolds as an optimization problem based on given weighted orientations of scaffolds and their pairs (e.g., coming from pair-end sequencing reads, long reads, or homologous relations). We formalize this problem using notion of a scaffold graph (i.e., a graph, where vertices correspond to the assembled contigs or scaffolds and edges represent connections between them). We prove that this problem is $\NP$-hard, and present a polynomial-time algorithm for solving its special case, where orientation of each scaffold is imposed relatively to at most two other scaffolds. We further develop an FPT algorithm for the general case of the OOS problem. 
\end{abstract}

\section{Introduction}
While genome sequencing technologies are constantly evolving, they are still unable to read at once complete genomic sequences from organisms of interest. Instead, they produce a large number of rather short genomic fragments, called \emph{reads}, originating from unknown locations and strands of the genome. The problem then becomes to assemble the reads into the complete genome. Existing genome assemblers usually assemble reads based on their overlap patterns and produce longer genomic fragments, called \emph{contigs}, which are typically interweaved with highly polymorphic and/or repetitive regions in the genome. Contigs are further assembled into \emph{scaffolds}, i.e., sequences of contigs interspaced with gaps.\footnote{We remark that contigs can be viewed as a special type of scaffolds with no gaps.} Assembling scaffolds into larger scaffolds (ideally representing complete chromosomes) is called the \emph{scaffold assembly problem}.

The scaffold assembly problem is known to be $\NP$-hard~\cite{Zimin2008,Kececioglu1995a,Gao2011,Pop2004a,Chen2017}, but there still exists a number of methods that use heuristic and/or exact algorithmic approaches to address it. 
The \emph{scaffold assembly problem} consists of two subproblems:
\begin{enumerate}
\item determine the order of scaffolds (\emph{scaffold order problem}); and
\item determine the orientation (i.e., strand of origin) of scaffolds (\emph{scaffold orientation problem}). 
\end{enumerate}

Some methods attempt to solve these subproblems jointly by using various types of additional data including jumping libraries \cite{Hunt2014a,Simpson2009,Koren2011,Gritsenko2012,LuoR.LiuB.XieY.LiZ.HuangW.YuanJ.andWang2012,Dayarian2010a,Boetzer2011}, long error-prone reads \cite{Warren2015,Bankevich2012a,Bashir2012,boetzer2014sspace,Lam2015}, homology relationships between genomes \cite{Assour2015,Aganezov2016,Avdeyev2016,Anselmetti2015,Kolmogorov2016}, etc. 
Other methods (typically based on wet-lab experiments \cite{Burton2013,Tang2015a,Nagarajan2008,Jiao2017,Putnam2016,Reyes-Chin-Wo2017}) can often reliably reconstruct the order of scaffolds, but may fail to impose their orientation.

The scaffold orientation problem is also known to be $\NP$-hard \cite{Bodily2015a,Kececioglu1995a}. 
Since the scaffold order problem can often be reliably solved with wet-lab based methods, this inspires us to consider the special case of the scaffold orientation problem with the given order of scaffolds, which we refer to as the \emph{orientation of ordered scaffolds} (OOS) problem.
We formulate the OOS as an optimization problem based on given weighted orientations of scaffolds and their pairs (e.g., coming from pair-end sequencing reads, long reads, or homologous relations). We prove that the OOS is $\NP$-hard both in the case of linear genomes and in the case of circular genomes. We present a polynomial-time algorithm for solving the special case of the OOS, where the orientation of each scaffold is imposed relatively to at most two other scaffolds, and further generalize it to an $\FPT$ algorithm for the general OOS problem. The proposed algorithms are implemented in the \CAMSA~\cite{Aganezov2017} software that have been developed for comparative analysis and merging of scaffold assemblies.

\section{Backgrounds}
We start with a brief description of the notation which have been used in \CAMSA framework and provides a unifying way to represent scaffold assemblies obtained by different methods. 

Let $\Scafs = \{s_i\}_{i=1}^n$ be the set of scaffolds. We represent an {\em assembly} of scaffolds from a set $\Scafs$ as a \emph{set} of \emph{assembly points}. Each assembly point is formed by an adjacency between two scaffolds from $\Scafs$. Namely, an assembly point $p = (s_i, s_j)$ tells that the scaffolds $s_i$ and $s_j$ are adjacent in the assembly. Additionally, we may know the orientation of either or both of the scaffolds and thus distinguish between three types of assembly points: 
\begin{enumerate} 
    \item $p$ is \emph{oriented} if the orientation of both scaffolds $s_i$ and $s_j$ is known;
    \item $p$ is \emph{semi-oriented} if the orientation of only one scaffold among $s_i$ and $s_j$ is known;
    \item $p$ is \emph{unoriented} if the orientation of neither of $s_i$ and $s_j$ is known.
\end{enumerate} 
We denote the known orientation of scaffolds in assembly points by overhead arrows, where the right arrow corresponds to the original genomic sequence representing a scaffold, while the left arrow corresponds to the reverse complement of this sequence. For example, $(\rd{s_i}, \ld{s_j})$, $(\rd{s_i}, s_j)$, and  $(s_i, s_j)$ are oriented, semi-oriented, and unoriented assembly points, respectively. We remark that assembly points $(\rd{s_i}, \rd{s_j})$ and $(\ld{s_j}, \ld{s_i})$ represent the same adjacency between oriented scaffolds; to make this representation unique we will require that in all assembly points $(s_i, s_j)$ we have $i<j$. Another way to represent the orientation of the scaffolds in an assembly point is by using superscripts $h$ and $t$ denoting the head and tail extremities of the scaffold's genomic sequence, e.g., $(\rd{s_i}, \rd{s_j})$ can also be written as $(s_i^h, s_j^t)$.

We will need an auxiliary function $\sn(p,i)$ defined on an assembly point $p$ and an index $i\in\{1, 2\}$ that returns the scaffold corresponding to the component $i$ of $p$ (e.g., $\sn((\rd{s_i}, \rd{s_j}), 2) = s_j$). We define a \emph{realization} of an assembly point $p$ as any oriented assembly point that can be obtained from $p$ by orienting the unoriented scaffolds. We denote the set of realizations of $p$ by $\rlz(p)$. If $p$ is oriented, than it has a single realization equal $p$ itself (i.e., $\rlz(p)=\{p\}$); if $p$ is semi-oriented, then it has two realizations (i.e., $|\rlz(p)| = 2$); and if $p$ is unoriented, then it has four realizations (i.e., $|\rlz(p)| = 4$). For example, 
\begin{equation}\label{eq:Rsjsj}
\rlz((s_i, s_j)) = \left\{(\rd{s_i}, \rd{s_j}), (\rd{s_i}, \ld{s_j}), (\ld{s_i}, \rd{s_j}), (\ld{s_i}, \ld{s_j})\right\}.
\end{equation}

An assembly point $p$ is called a \emph{refinement} of an assembly point $q$ if $\rlz(p)\subset\rlz(q)$. From now on, we assume that no assembly point in a given assembly is a refinement of another assembly point (otherwise we simply discard the latter assembly point as less informative). We further assume that in a given assembly there are no assembly points $(\rd{s_i}, s_j)$ and $(s_i, \rd{s_j})$ such that either $s_i$ or $s_j$ belongs to an assembly point differing from $(\rd{s_i}, s_j)$ and $(s_i, \rd{s_j})$  (otherwise we simply replace $(\rd{s_i}, s_j)$ and $(s_i, \rd{s_j})$ on $(\rd{s_i}, \rd{s_j})$\footnote{It will be seen later that the any assembly realization in this case is conflicting.}). Similarly, we assume that there are no assembly points $s_i$, $s_j$ that forms $(\rd{s_i}, \ld{s_j}), (\ld{s_i}, \rd{s_j}), (\ld{s_i}, \ld{s_j})$. We refer to an assembly with no assembly point refinements and no pairs described above as a \emph{proper assembly}.

For a given assembly $\Ass$ we will use subscripts \emph{u}/\emph{s}/\emph{o} to denote the sets of unoriented/semi-oriented/oriented assembly points in $\Ass$ (e.g., $\Ass_u\subset\Ass$ is the set of all unoriented assembly points from $\Ass$). We also denote by $\Scafs(\Ass)$ the set of scaffolds appearing in the assembly points from $\Ass$.

We call two assembly points \emph{overlapping} if they involve the same scaffold, and further call them \emph{conflicting} if they involve the same extremity of this scaffold. We generalize this notion for semi-oriented and unoriented assembly points: two assembly points $p$ and $q$ are \emph{conflicting} if all pairs of their realizations $\{\ar{p}, \ar{q}\}\in\rlz(p)\times\rlz(p)$ are conflicting. If some, but not all, pairs of the realizations are conflicting, $p$ and $q$ are called \emph{semi-conflicting}. Otherwise, $p$ and $q$ are called \emph{non-conflicting}.

We extend the notion of non-/semi- conflictness to entire assemblies as follows. A scaffold assembly $\Ass$ is \emph{non-conflicting} if all pairs of assembly points in it are non-conflicting, and $\Ass$ is \emph{semi-conflicting} if all pairs of assembly points are non-conflicting or semi-conflicting with at least one pair being semi-conflicting. 

\section{Methods} 
\subsection{Assembly Realizations}

For an assembly $\Ass = \{p_i\}_{i=1}^k$, an assembly $\ar{\Ass} = \{q_i\}_{i=1}^k$ is called a \emph{realization}\footnote{It can be easily seen that a realization of $\Ass$ may exist only if $\Ass$ is proper.} of $\Ass$ if there exists a permutation $\pi$ of order $k$ such that $q_{\pi_i}\in \rlz(p_i)$ for all $i=1,2,\dots,k$. We denote by $\rlz(\Ass)$ the set of realizations of assembly $\Ass$, and by $\ncrlz(\Ass)$ the set of non-conflicting realizations among them. 

We define the \emph{scaffold assembly graph} $\SAG(\Ass)$ on the set of vertices $\{s^h, s^t\ :\ s\in\Scafs(\Ass)\}$ and edges of two types: directed edges $(s^t, s^h)$ that encode scaffolds from $\Scafs(\Ass)$, and undirected edges that encode all possible realizations of all assembly points in $\Ass$ (Fig.~\ref{fig:sag_og_cog}a). We further define the \emph{order (multi)graph} $\OG(\Ass)$ formed by the set of vertices $\Scafs(\Ass)$ and the set of undirected edges $\{\{\sn(p,1), \sn(p,2)\}\ :\ p\in\Ass\}$ (Fig.~\ref{fig:sag_og_cog}b). The order graph can also be obtained from $\SAG(\Ass)$ by first contracting the directed edges, and then by substituting all edges that encode realizations of the same assembly point with a single edge (Fig.~\ref{fig:sag_og_cog}b). We define the \emph{contracted order graph} $\COG(\Ass)$ obtained from $\OG(\Ass)$ by replacing all multi-edges edges with single edges (Fig.~\ref{fig:sag_og_cog}c).

Let $\deg(v)$ be the degree of a vertex $v$, i.e., the number of edges (counted with multiplicity) incident to $v$ in $\OG(\Ass)$. We call the order graph $\OG(\Ass)$ \emph{non-branching} if $\deg(v)\leq 2$ for all vertices $v$ of $\OG(\Ass)$. 

\begin{figure}
    \includegraphics[width=\textwidth]{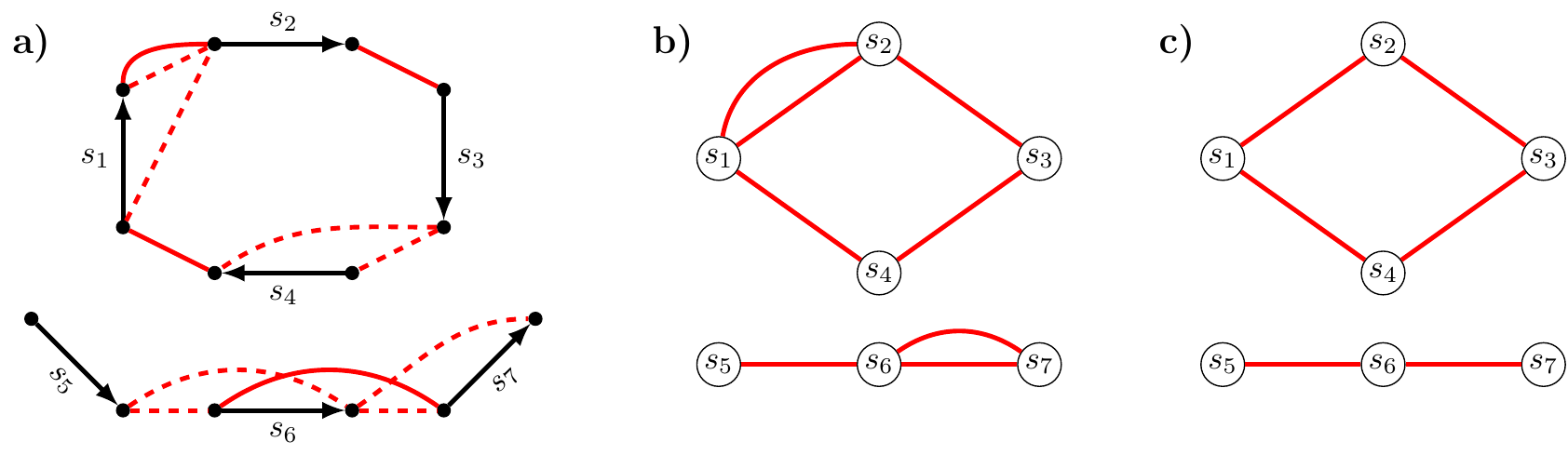}
    \caption{For an assembly $A = \{(s_1, \rd{s_2}), (\rd{s_1}, \rd{s_2}), (\rd{s_2}, \rd{s_3}), (\rd{s_3}, s_4), (\ld{s_1}, \ld{s_4}), (\rd{s_5}, s_6)$, $(\ld{s_6}, \rd{s_7}), (\rd{s_6}, s_7)\}$, \textbf{a)} the scaffold assembly graph $\SAG(A)$, where semi-oriented assembly points, oriented assembly points, and scaffolds are represented by dashed red edges, solid red edges, and directed black edges, respectively. \textbf{b)} The order graph $\OG(A)$. \textbf{c)} The contracted order graph $\COG(A)$.}
    \label{fig:sag_og_cog}
\end{figure}

\begin{lem}
\label{lem:non-branching}
For a non-conflicting realization $\ar{\Ass}$ of an assembly $\Ass$, $\OG(\ar{\Ass})$ is non-branching. 
\end{lem}
\begin{proof} 
Each vertex $v$ in $\OG(\ar{\Ass})$ represents a scaffold, which has two extremities and thus can participate in at most two non-conflicting assembly points in $\ar{\Ass}$. Hence, $\deg(v)\leq 2$.
\end{proof}

We notice that any non-conflicting realization $\ar{\Ass}$ of an assembly $\Ass$ provides orientation for all scaffolds involved in each connected component of $\SAG(\ar{\Ass})$ (as well as of $\OG(\ar{\Ass})$ and $\COG(\ar{\Ass})$) relatively to each other. 

\begin{thm}\label{thm:NR1}
An assembly $\Ass$ has at least one non-conflicting realization (i.e., $|\ncrlz(\Ass)|\ge 1$) if and only if $\Ass$ is non-conflicting or semi-conflicting and $\OG(\Ass)$ is non-branching.
\end{thm}    
\begin{proof}
Suppose that $|\ncrlz(\Ass)|\ge 1$ and pick any $\ar{\Ass}\in \ncrlz(\Ass)$. Then for every pair of assembly points $p,q\in \Ass$, their realizations in $\ar{\Ass}$ are non-conflicting, implying that $p$ and $q$ are either non-conflicting or semi-conflicting. Hence, $\Ass$ is non-conflicting or semi-conflicting. Since $\Ass$ is a proper assembly, we have $\OG(\Ass)=\OG(\Ass')$. Taking into the account that $\Ass'$ is non-conflicting, Lemma~\ref{lem:non-branching} implies that $\OG(\Ass)$ is non-branching.
         
Vice versa, suppose that $\Ass$ is non-conflicting or semi-conflicting and $\OG(\Ass)$ is non-branching. To prove that $|\ncrlz(\Ass)|\ge 1$, we will orient unoriented scaffolds in all assembly points in $\Ass$ without creating conflicts. Every scaffold $s$ corresponds to a vertex $v$ in $\OG(\Ass)$ of degree at most $2$. If $\deg(v)=1$, then $s$ participates in one assembly point $p$, and $s$ is either already oriented in $p$ or we pick an arbitrary orientation for it. If $\deg(v)=2$, then $s$ participates in two overlapping assembly points $p$ and $q$. If $s$ is not oriented in either of $p$, $q$, we pick an arbitrary orientation for it consistently  across $p$ and $q$ (i.e., keeping them non-conflicting). If $s$ is oriented in exactly one assembly point, we orient the unoriented instance of $s$ consistently with its orientation in the other assembly point. 
Since conflicts may appear only between assembly points that share a vertex in $\OG(\Ass)$, the constructed orientations produce no new conflicts. On other hand, the scaffolds that are already oriented in $\Ass$ impose no conflicts since $\Ass$ is non-conflicting or semi-conflicting. Hence, the resulting oriented assembly points form a non-conflicting assembly from $\ncrlz(\Ass)$, i.e., $|\ncrlz(\Ass)|\ge 1$.
\end{proof}

We remark that if $\OG(\Ass)$ is branching, the assembly $\Ass$ may be semi-conflicting but have $|\ncrlz(\Ass)|=0$. An example is given by $\Ass=\{(s_1,s_{i+1})\}_{i=1}^k$ with $k>2$, which contains no conflicting assembly points (in fact, all assembly points in $\Ass$ are semi-conflicting), but $|\ncrlz(\Ass)|=0$. 

From now on, we will always assume that assembly $\Ass$ has at least one non-conflicting realization (i.e., $|\ncrlz(\Ass)|\ge 1$). For an assembly $\Ass$, the orientation of some scaffolds from $\Scafs(\Ass)$ does not depend on the choice of a realization from $\ncrlz(\Ass)$ (we denote the set of such scaffolds by $\Scafs_o(\Ass)$), while the orientation of other scaffolds within some assembly points varies across realizations from $\ncrlz(\Ass)$ (we denote the set of such scaffolds by $\Scafs_u(\Ass)$). Trivially, we have $\Scafs_u(\Ass)\cup \Scafs_o(\Ass)=\Scafs(\Ass)$. It can be easily seen that the set $\Scafs_u(\Ass)$ is formed by the scaffolds for which the orientation in the proof of Theorem~\ref{thm:NR1} was chosen arbitrarily, implying the following statement.
    
\begin{cor}
For a given assembly $\Ass$ with $|\ncrlz(\Ass)|\ge 1$, we have $|\ncrlz(\Ass)|=2^{|\Scafs_u(\Ass)|}$.
\end{cor}
        
We label scaffolds from $\Scafs(\Ass)$ with integers $\{1,\dots,|\Scafs(\Ass)|\}$. From computational perspective, we assume that we can get a scaffold from its name and vice verca in $\bigO{1}$ time.

\begin{lem}
\label{lem:non_conflict_time}
Testing whether a given assembly $\Ass$ has a non-conflicting realization can be done in $\bigO{k}$ time, where $k=|\Scafs(\Ass)|$.
\end{lem}
\begin{proof} To test whether $\Ass$ has a non-conflicting realization, we first create a hash table indexed by $\Scafs(\Ass)$ that for every scaffold $s\in \Scafs(\Ass)$ will contain a list of assembly points that involve $s$. We iterate over all assembly points $p\in\Ass$ and add $p$ to two lists in the hash table indexed by the scaffolds participating in $p$. If the length of some list becomes greater than 2, then $\Ass$ is conflicting and we stop. If we successfully complete the iterations, then every scaffold from $\Scafs(\Ass)$ participates in at most two assembly points in $\Ass$, and thus we made $\bigO{k}$ steps of $\bigO{1}$ time each. 
    
Next, for every scaffold whose list in the hash table has length 2, we check whether the corresponding assembly points are either non-conflicting or semi-conflicting. If not, then $\Ass$ is conflicting and we stop. If the check completes successfully, then $\Ass$ has a non-conflicting realization by Theorem~\ref{thm:NR1}. The check takes $\bigO{k}$ steps of $\bigO{1}$ time each, and thus the total running time comes to $\bigO{k}$. 
\end{proof}

A pseudocode for the test described in the proof of Lemma~\ref{lem:non_conflict_time} is given Algorithm~\ref{algo:checkNCR} in the Appendix.

\begin{lem}\label{lem:compSu}
For a given assembly $\Ass$ with $|\ncrlz(\Ass)|\geq 1$, the set $\Scafs_u(\Ass)$ can be computed in $\bigO{k}$ time, where $k=|\Scafs(\Ass)|$.
\end{lem}    
\begin{proof}
We will construct the set $S = \Scafs_u(\Ass)$ iteratively.
Initially we let $S=\emptyset$. Following the algorithm described in the proof for Lemma~\ref{lem:non_conflict_time}, we construct a hash table  that for every scaffold $i\in\Scafs(\Ass)$ contains a list of assembly points that involve $i$ (which takes $\bigO{k}$ time). Then for every $i\in\Scafs(\Ass)$, we check if either of the corresponding assembly points provides an orientation for $i$; if not, we add $i$ to $S$. This check for each scaffolds takes $\bigO{1}$ time, bringing the total running time to $\bigO{k}.$  
\end{proof}
    
A pseudocode for the computation of $\Scafs_u(\Ass)$ described in the proof of Lemma~\ref{lem:compSu} is given in Algorithm~\ref{algo:computingSu} in the Appendix.

\subsection{Problem Formulations}  
\paragraph{Orientation of Ordered Scaffolds}  
For a non-conflicting assembly $\Ass$ composed only of oriented assembly points, an assembly point $p$ on scaffolds $s_i, s_j\in\Scafs(\Ass)$ has a \emph{consistent orientation with $\Ass$} if for some $\ar{p}\in\rlz(p)$ there exists a path connecting edges $s_i$ and $s_j$ in $\SAG(\Ass)$ such that direction of edges $s_i$ and $s_j$ at the path ends is consistent with $\ar{p}$ (e.g., in Fig.~\ref{fig:sag_og_cog}a, the assembly point $(\rd{s_1}, \rd{s_3})$ has a consistent orientation with the assembly $\Ass$). Furthermore, for a non-conflicting assembly $\Ass$ that has at least one non-conflicting realization, an assembly point $p$ has a \emph{consistent orientation with $\Ass$} if $\ar{p}$ has a consistent orientation with $\ar{\Ass}$ for some $\ar{p}\in\rlz(p)$ and $\ar{\Ass} \in \ncrlz(\Ass)$.

We formulate the orientation of ordered scaffolds problem as follows.
    
\begin{problem}[Orientation of Ordered Scaffolds, OOS]
Let $\Ass$ be an assembly and $\Orp$ be a set\footnote{More generally, $\Orp$ may be a multiset whose elements have real positive multiplicities (weights).} of assembly points such that $|\ncrlz(\Ass)|\geq 1$ and $\Scafs(\Orp)\subset\Scafs(\Ass)$. Find a non-conflicting realization $\ar{\Ass}\in\ncrlz(\Ass)$ that maximizes the number (total weight) of assembly points from $\Orp$ having consistent orientations with $\ar{\Ass}$.
\end{problem}
    
From the biological perspective, the OOS can be viewed as a formalization of the case where (sub)orders of scaffolds have been determined (which defines $\Ass$), while there exists some information (possibly coming from different sources and conflicting) about their relative orientation (which defines $\Orp$). The OOS asks to orient unoriented scaffolds in the given scaffold orders in a way that is most consistent with the given orientation information.
    
We also remark that the OOS can be viewed as a fine-grained variant of the scaffold orientation problem studied in~\cite{Bodily2015a}. In our terminology, the latter problem concerns an artificial circular genome $\Ass$ formed by the given scaffolds in an arbitrary order (so that there is a path connecting any scaffold or its reverse complement to any other scaffold in $\OG(\Ass)$), and $\Orp$ formed by unordered pairs of scaffolds supplemented with the binary information on whether each such pair come from the same or different strands of the genome. In contrast, in the OOS, the assembly $\Ass$ is given and $\OG(\Ass)$ does not have to be connected or non-branching, while $\Orp$ may provide a pair of scaffolds with up to four options (as in \eqref{eq:Rsjsj}) of their relative orientation.
    
\paragraph{Non-branching Orientation of Ordered Scaffolds}
At the latest stages of genome assembly, the constructed scaffolds are usually of significant length. If (sub)orders for these scaffolds are known, it is rather rare to have orientation-imposing information that would involve non-neighboring scaffolds. Or, more generally, it is rather rare to have orientation imposing information for one scaffold with respect to more than two other scaffolds. This inspires us to consider a special case of the OOS problem:
    
\begin{problem}[Non-branching Orientation of Ordered Scaffolds, NOOS]
Given an OOS instance $(\Ass,\Orp)$ such that the graph $\COG(\Orp)$ is non-branching. Find $\ar{\Ass}\in\ncrlz(\Ass)$ that maximizes the number of assembly points from $\Orp$ having consistent orientations with $\ar{\Ass}$.
\end{problem}
   
\subsection{$\NP$-hardness of the OOS}
We consider two important partial cases of the OOS, where the assembly $\Ass$ represents a linear or circular genome up to unknown orientations of the scaffolds. In these cases, the graph $\OG(\Ass)$ forms a collection of paths or cycles, respectively. Below we prove that the OOS in both these cases is $\NP$-{\rm hard}.
    
\begin{lem}\label{thm:OOSlin}
The OOS for linear genomes is $\NP$-{\rm hard}.
\end{lem}
\begin{proof}
We will construct a polynomial-time reduction from the $\MAX\ 2$-$\DNF$ problem, which is known to be $\NP$-hard~\cite{Bazgan2003,Escoffier2005}. Given an instance $I$ of $\MAX\ 2$-$\DNF$ consisting of conjunctions $C = \{c_i\}_{i=1}^k$ on variables $X = \{x_i\}_{i=1}^n$, we define an assembly 
$$\Ass = \{ (0,x_1) \} \cup \{ (x_i, x_{i+1})\ :\ i=1,2,\dots,n-1\}.$$
We then construct a set of assembly points $\Orp$ from the clauses in $C$ as follows. For each clause $c\in C$ with two variables $x_i$ and $x_j$ ($i<j$), we add an oriented assembly point on scaffolds $x_i, x_j$ to $\Orp$ with the orientation depending on the negation of these variables in $c$ (i.e., a clause $x_i\wedge\overline{x_j}$ is translated into an assembly point $(\rd{x_i}, \ld{x_j})$). For each clause from $C$ with a single variable $x$, we add an assembly point $(\rd{0}, \rd{x})$ or $(\rd{0}, \ld{x})$ depending whether $x$ is negated in the clause. 
        
It is easy to see that the constructed assembly $\Ass$ is semi-conflicting and $\OG(\Ass)$ is a path, and thus by Theorem~\ref{thm:NR1} $\Ass$ has a non-conflicting realization. Hence, $\Ass$ and $\Orp$ form an instance of the OOS for linear genomes. A solution $\ar{\Ass}$ to this OOS provides an orientation for each $x\in \Scafs$ that maximizes the number of assembly points from $\Orp$ having consistent orientations with $\ar{\Ass}$. A solution to $I$ is obtained from $\ar{\Ass}$ as the assignment of $0$ or $1$ to each variable $x$ depending on whether the orientation of scaffold $x$ in $\ar{\Ass}$ is forward or reverse. Indeed, since each assembly point in $\Orp$ having consistent orientation with $\ar{\Ass}$ corresponds to a truthful clause in $I$, the number of such clauses is maximized. 
        
It is easy to see that the OOS instance and the solution to $I$ can be computed in polynomial time, thus we constructed a polynomial-time reduction from the $\MAX\ 2$-$\DNF$ to the OOS for linear genomes.        
\end{proof}
    
\begin{lem}\label{thm:OOScir}
The OOS for circular genomes is $\NP$-{\rm hard}.
\end{lem}
\begin{proof}
We construct a polynomial-time reduction from the $\MAX$-$\CUT$ problem, which is known to be $\NP$-hard~\cite{garey1976some,garey1979computers}. An instance $I$ of $\MAX$-$\CUT$  for a given a graph $(V, E)$ asks to partition the set of vertices $V = \{v_i\}_{i=1}^n$ into two disjoint subsets $V_1$ and $V_2$ such that the number of edges $\{u, v\}\in E$ with $u\in V_1$ and $v\in V_2$ is maximized. For a given instance $I$ of $\MAX$-$\CUT$ problem, we define the assembly 
$$\Ass = \left\{ (v_i,v_{i+1})\ :\ i=1,2,\dots,n-1\right\} \cup \left\{ (v_1,v_n) \right\}$$
and the set of assembly points 
$$\Orp = \left\{ (\rd{v_i}, \ld{v_j})\ :\ \{v_i, v_j\}\in E \right\}.$$
It is easy to see that $\Ass$ has a non-conflicting realization and $\OG(\Ass)$ is a cycle, i.e., $\Ass$ and $\Orp$ form an instance of the OOS for circular genomes. A solution $\ar{\Ass}$ to this OOS instance provides orientations for all elements $\Scafs(\Ass)=V$ that maximizes the number of assembly points from $\Orp$ having consistent orientations with $\ar{\Ass}$. A solution to $I$ is obtained as the partition of $V$ into two disjoint subsets, depending on the orientation of scaffolds in $\ar{\Ass}$ (forward vs reverse). Indeed, since each assembly point in $\Orp$ having a consistent orientation with $\ar{\Ass}$ corresponds to an edge from $E$ whose endpoints belong to distinct subsets in the partition, the number of such edges is maximized. 
        
It is easy to see that the OOS instance and the solution to $I$ can be computed in polynomial time, thus we constructed a polynomial-time reduction from the $\MAX$-$\CUT$ to the OOS for circular genomes.
\end{proof}
    
As a trivial consequence of Lemmas~\ref{thm:OOSlin} and \ref{thm:OOScir}, we obtain that the general OOS problem is $\NP$-{\rm hard}.
    
\begin{thm}
The OOS is $\NP$-hard.
\end{thm}
    
\subsection{Properties of the OOS}
In this subsection, we formulate and prove some important properties of the OOS.

\subsubsection*{Connected Components of $\OG(\Ass)$}

Below we show that an OOS instance can also be solved independently for each connected component of $\OG(\Ass)$. We start with the following lemma that trivially follows from the definition of consistent orientation.

\begin{lem}\label{lem:cons_conn}
Let $\Ass$ be an assembly such that $|\ncrlz(\Ass)|\geq 1$. An assembly point on scaffolds $s_i, s_j\in\Scafs(\Ass)$ may have a consistent orientation with $\Ass$ only if both $s_i$ and $s_j$ belong to the same connected component in $\OG(\Ass)$.
\end{lem}

\begin{thm}\label{thm:splitOGA}
Let $(\Ass,\Orp)$ be an OOS instance, and $\Ass = \Ass_1 \cup \dots \cup \Ass_k$ be the partition such that $\OG(\Ass_1),\dots,\OG(\Ass_k)$ represent the connected components of $\OG(\Ass)$. For each $i=1,2,\dots,k$, define $\Orp_i = \{ p\in\Orp\ :\ \sn(p,1),\sn(p,2)\in\Scafs(\Ass_i)\}$ and let $\ar{\Ass}_i$ be a solution to the OOS instance $(\Ass_i,\Orp_i)$. Then $\ar{\Ass}_1\cup\dots\cup\ar{\Ass}_k$ is a solution to the OOS instance $(\Ass,\Orp)$.
\end{thm}
\begin{proof}
Lemma~\ref{lem:cons_conn} implies that we can discard from $\Orp$ all assembly points that are formed by scaffolds from different connected components in $\OG(\Ass)$. Hence, we may assume that $\Orp = \Orp_1\cup\dots\cup\Orp_k\}$.

Lemma~\ref{lem:cons_conn} further implies that an assembly point from $\Orp_i$ may have a consistent orientation with $\Ass_j$ only if $i=j$. Therefore, any solution to the OOS instance $(\Ass,\Orp)$ is formed by the union of solutions to the OOS instances $(\Ass_i,\Orp_i)$.
\end{proof}

Theorem~\ref{thm:splitOGA} allows us focus on instances of the OOS, where $\OG(\Ass)$ is connected and thus forms a path or a cycle (by Theorem~\ref{thm:NR1}). 

\subsubsection*{Connected Components of $\OG(\Orp)$}

Below we show that an OOS instance can also be solved independently for each connected component of $\OG(\Orp)$. We need the following lemma that trivially holds.
        
\begin{lem}\label{lem:consistent}
Let $\Ass$ be an assembly such that $|\ncrlz(\Ass)|\geq 1$, and $s_i,s_j$ be scaffolds from the same connected component $C$ in $\OG(\Ass)$. Then an unoriented assembly point $(s_i,s_j)$ has a consistent orientation with $\Ass$. Furthermore, if $C$ is a cycle, then any semi-oriented assembly point on $s_i,s_j$ has a consistent orientation with $\Ass$. 
\end{lem}
    
By Lemma~\ref{lem:consistent}, we can assume that $\Orp$ does not contain any unoriented assembly points (i.e., $\Orp = \Orp_o \cup \Orp_s$). Furthermore, if $\OG(\Ass)$ is a cycle, we can assume that $\Orp=\Orp_o$ (i.e., $\Orp$ consists of oriented assembly points only). We consider two cases depending on whether $\OG(\Ass)$ forms a path or a cycle.
        
\paragraph*{\textbf{$\OG(\Ass)$ is a path}} 
Suppose that $\OG(\Ass) = (s_1,s_2,\dots,s_n)$ is a path and $\Orp = \Orp_o \cup \Orp_s$. Let ${\cal C}$ be the set of connected components of $\OG(\Orp)$. 

Consider any $C \in {\cal C}$. Let $(s_{j_1},\dots,s_{j_{m}})$ be a vertex sequence of $C$ such that $j_{1}<j_{2}<\dots<j_{m}$, where $m$ is the number of vertices in $C$. We define an assembly $\Ass_{C}$ such that $\OG(\Ass_{C})$ is the path $(x,s_{j_1},\dots,s_{j_{m}},y)$, where $x$ and $y$ are artificial vertices, and the assembly points in $\Ass_{C}$ (corresponding to the edges in $\OG(\Ass_{C})$) are oriented or semi-oriented as follows.

\begin{itemize}
\item The edges $\{x,s_{j_{1}}\}$ and $\{s_{j_{m}},y\}$ correspond to semi-oriented assembly points $(\rd{x},s_{j_{1}})$ and $(s_{j_{m}},\rd{y})$, respectively;
\item For any $l \in \{1,\dots,m - 1\}$, the assembly point corresponding to the edge $\{s_{j_{l}},s_{j_{l+1}}\}$ is inherited from the assembly points corresponding to edges $\{s_{j_{l}},s_{j_{l}+1}\}$ and $\{s_{j_{l+1}-1},s_{j_{l+1}}\}$.
\end{itemize}

We further define $\Orp_C$ as a set formed by the assembly points from $C$ and the following assembly points. For each semi-oriented assembly point $p\in\Orp$ formed by scaffolds $s_m$ and $s_l$ ($m<l$), $\Orp_C$ contains: 
\begin{itemize}
    \item an oriented point $p'$ formed by $s_m$ and $\rd{y}$ whenever $s_m$ is oriented in $p$ and belongs to $C$ (and its orientation in $p'$ is inherited from $p$);
    \item an oriented point $p''$ formed by $\rd{x}$ and $s_l$ whenever $s_l$ is oriented in $p$ and belongs to $C$ (and its orientation in $p''$ is inherited from $p$) (Fig.~\ref{fig:og_o_cc_decom}).
\end{itemize}

Now, we assume that $\Ass_C$ and $\Orp_C$ ($C \in {\cal C}$) are defined as above for the OOS instance $(\Ass,\Orp)$. For each $C \in {\cal C}$, let $\ar{\Ass}_C$ be a solution to the OOS instance $(\Ass_C,\Orp_C)$. We construct a non-conflicting realization $\ar{\Ass}\in \ncrlz(\Ass)$ as follows:
\begin{itemize}
    \item for a scaffold $s$ \emph{present} in some $\ar{\Ass}_C$, $\ar{\Ass}$ inherits the orientation of $s$ from $\ar{\Ass}_i$;
    \item for a scaffold $s$ \emph{not present} in any $\ar{\Ass}_C$, $\ar{\Ass}$ inherits the orientation of $s$ from $\Ass$ if $s$ is oriented in any assembly point of $\Ass$, or otherwise has $s$ arbitrarily oriented.
\end{itemize}
The following theorem shows that constructed $\ar{\Ass}$ is a solution to the OOS instance $(\Ass,\Orp)$.

\begin{figure}[!t]
        \centering
        \includegraphics[width=\textwidth]{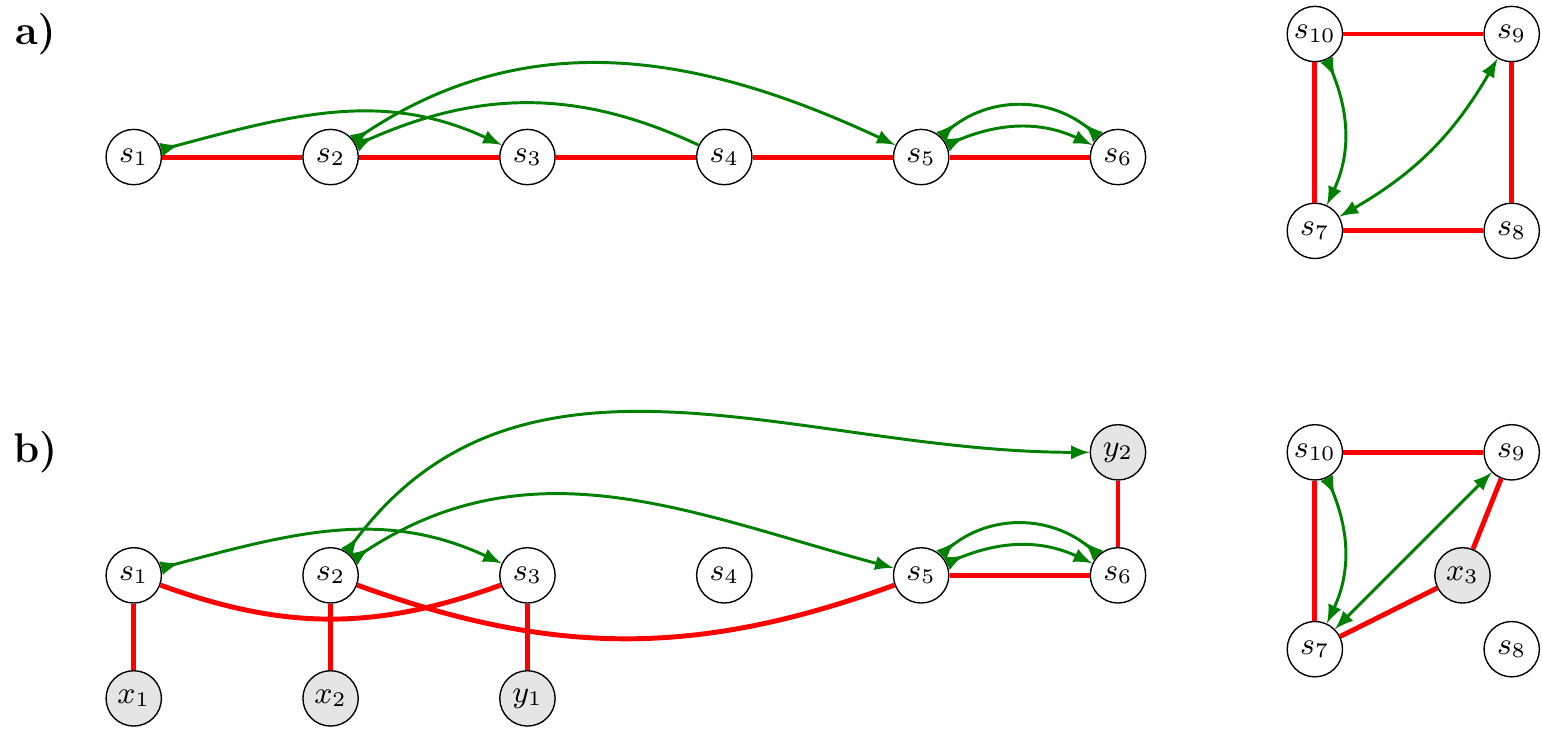}
        \caption{Decomposition of an OOS problem instance $(\Ass, \Orp)$ based on the connected components of $\OG(\Orp_o)$. 
          \textbf{a)} The superposition of $\OG(\Ass)$ (red edges) and $\OG(\Orp)$ (green edges), where arrows (if present) at the ends of green edges encode the orientation of the scaffolds in the corresponding assembly points. 
          \textbf{b)} The superposition of five graphs $\OG(\Ass_i)$ (red edges) and three graphs $\OG(\Orp_j)$ (green edges) constructed based on the connected components of $\OG(\Orp_o)$. 
          Unless $\OG(\Ass_i)$ is formed by an isolated vertex, it contains artificial vertices $x_i$ and $y_i$, which coincide if $\OG(\Ass_i)$ is a cycle.}
        \label{fig:og_o_cc_decom}
\end{figure}
    
\begin{thm}\label{thm:splitOGO}
Let $(\Ass,\Orp)$ be an OOS instance, and $\ar{\Ass}\in \ncrlz(\Ass)$ be defined as above. Then $\ar{\Ass}$ is a solution to the OOS instance $(\Ass,\Orp)$.
\end{thm}
\begin{proof}
The graph $\SAG(\ar{\Ass})$ can be viewed as an ordered sequence of directed scaffold edges (interweaved with undirected edges encoding assembly points). Then each $\SAG(\ar{\Ass}_i)$, with the exception of scaffold edges $x_i$ and $y_i$, corresponds to a subsequence of this sequence.

Each oriented assembly point $p\in\Orp$ is formed by scaffolds $u,v$ from $C_i$ for some $i\in\{1,\dots,k\}$. Then $p\in\Orp\cap\Orp_i$ and there exist a unique path in $\SAG(\ar{\Ass}_i)$ and a unique path in $\SAG(\ar{\Ass})$ having the same directed edges $u,v$ at the ends. Hence, if $p$ has a consistent orientation with one of assemblies $\ar{\Ass}$ or $\ar{\Ass}_i$, then it has a consistent orientation with the other.

Each semi-oriented assembly point $p\in\Orp$ formed by scaffold $u,v$ corresponds to an oriented assembly point $q\in\Orp_i$ (for some $i$) formed by $u$ and $y_i$ (in which case $u\in C_i$ and $u$ is oriented in $p$), or by $x_i$ and $v$ (in which case $v\in C_i$ and $v$ is oriented in $p$). Without loss of generality, we assume the former case. Then there exists a unique path $Q$ in $\SAG(\ar{\Ass}_i)$ connecting directed edges $u$ and $y_i$, and there exists a unique path $P$ in $\SAG(\ar{\Ass})$ connecting directed edges $u$ and $v$, where the orientation of $u$ is the same in the two paths. By construction, the orientation of $y_i$ in $q$ matches that in $Q$. Hence, 
$q$ has a consistent orientation with $\ar{\Ass}_i$ if and only if the orientation of $u$ in $q$ matches that in $Q$, which happens if and only if the orientation of $u$ in $p$ matches its orientation in $P$, i.e., $p$ has a consistent orientation with $\ar{\Ass}$. 
We proved that the number of assembly points from $\Orp$ having consistent orientation with $\ar{\Ass}$ equals the total number of assembly points from $\Orp_i$ having consistent orientation with $\ar{\Ass}_i$ for all $i=1,2,\dots,k$. It remains to notice that this number is maximum possible, i.e., $\ar{\Ass}$ is indeed a solution to the OOS instance $(\Ass,\Orp)$ (if it is not, then the sets $\Ass_i$ constructed from $\Ass$ being an actual solution to the OOS will give a better solution to at least one of the subproblems).
\end{proof}

\paragraph*{\textbf{$\OG(\Ass)$ is a cycle}} 
In this case, we can construct subproblems based on the connected components of $\OG(\Orp)$ similarly to Case 1, with the following differences. First, by Lemma~\ref{lem:consistent}, we assume that $\Orp=\Orp_o$ (discarding all unoriented and semi-oriented assembly points from $\Orp$).
Second, we assume that $x_i=y_i$ and thus $\OG(\Ass_i)$ forms a cycle. Theorem~\ref{thm:splitOGO} still holds in this case.

\subsubsection*{Articulation vertices in $\OG(\Orp)$}
While Theorem~\ref{thm:splitOGO} allows us to divide the OOS problems into subproblems based on the connected components of $\OG(\Orp)$, we show below that similar division is possible when $\OG(\Orp)$ is connected but contains an articulation vertex.\footnote{We remind that a vertex is \emph{articulation} if its removal from the graph increases the number of connected components.} 

A vertex $v$ in $\OG(\Orp)$ (or in $\COG(\Orp)$) is called \emph{oriented} if $v\in \Scafs_o(\Ass)$. Otherwise, $v$ is called \emph{unoriented}. Let $(\Ass, \Orp)$ be an instance of the OOS problem such that both $\OG(\Ass)$ and $\OG(\Orp)$ are connected. Let $v$ be an oriented articulation vertex in $\OG(\Orp)$, defining a partition of $\Scafs(\Orp)$ into disjoint subsets:

\begin{equation}\label{eq:partV}
\Scafs(\Orp) = \{v\} \cup V_1 \cup V_2 \cup \dots \cup V_k,
\end{equation}
where $k>1$ and the $V_i$ represent the vertex sets of the connected components resulted from removal of $v$ from $\OG(\Orp)$. To divide the OOS instance $(\Ass, \Orp)$ into subinstances, we construct a new OOS instance $(\hat\Ass, \hat\Orp)$ as follows. 

We introduce copies $v_1, \dots, v_k$ of $v$, and construct $\hat\Ass$ from $\Ass$ by replacing a path $(u,v,w)$ in $\OG(\Ass)$ with a path $(u,v_1,v_2,\dots,v_k,w)$ where all $v_i$ inherit the orientation from $v$. Then we construct $\hat\Orp$ from $\Orp$ by replacing in each assembly point $p$ formed by $v$ and $u\in V_i$ (for some $i\in\{1,2,\dots,k\})$ with an assembly point formed by $v_i$ and $u$ (keeping their orientations intact). 


The OOS instance $(\hat\Ass, \hat\Orp)$ enables application of Theorem~\ref{thm:splitOGO}. Indeed, by construction, the vertex sets of the connected components of $\OG(\hat\Orp)$ are $\{v_i\} \cup V_i$, where $i \in \{1,2,\dots,k\}$. Hence, by Theorem~\ref{thm:splitOGO} the OOS instance $(\hat\Ass, \hat\Orp)$ can solved by dividing into OOS subinstances corresponding to the connected components of $\OG(\hat\Orp)$.  

Now, we assume that we have a solution to the OOS instance $(\hat\Ass, \hat\Orp)$. We construct a non-conflicting realization $\ar{\Ass}\in \ncrlz(\Ass)$ from a solution to the OOS instance $(\hat\Ass, \hat\Orp)$ by replacing every scaffold $v_i$ with $v$. 

The following theorem shows that the constructed $\ar{\Ass}$ is a solution to the OOS instance $(\Ass,\Orp)$.

\begin{thm}\label{thm:atr_vert}
Let $(\Ass, \Orp)$ be an OOS instance such that both $\OG(\Ass)$ and $\OG(\Orp)$ are connected, and $\ar{\Ass}$ be defined as above. Then $\ar{\Ass}$ is a solution to the OOS instance $(\Ass,\Orp)$.

\end{thm}
\begin{proof}
Let $\ar{\hat\Ass}$ be a solution to the OOS instance $(\hat\Ass, \hat\Orp)$, and $\ar{\Ass}$ be obtained from $\ar{\hat\Ass}$ by replacing every $v_i$ with $v$. We remark that $\Orp$ can be obtained from $\hat\Orp$ by similar replacement.
    
This establishes an one-to-one correspondence between the assembly points in $\ar{\hat\Ass}$ and $\ar{\Ass}$, as well as between the assembly points in $\ar{\hat\Orp}$ and $\ar{\Orp}$. It remains to show that consistent orientations are invariant under this correspondence.
    
We remark that $\SAG(\ar{\Ass})$ can be obtained from $\SAG(\ar{\hat\Ass})$ by replacing a sequence of edges $(r_1,v_1,r_2,v_2,\dots,r_k,v_k,r_{k+1})$, where $r_i$ are assembly edges, with a sequence of edges $(r_1,v,r_2)$. Therefore, if there exists a path in one graph proving existence of consistent orientation for some assembly point, then there exists a corresponding path in the other graph (having the same orientations of the end edges).
\end{proof}
    
\subsection{Algorithms for the NOOS and the OOS}
In this section, by Theorems~\ref{thm:splitOGA} and \ref{thm:splitOGO}, we can assume that both $\OG(\Ass)$ and $\OG(\Orp)$ are connected.

\subsubsection*{A Polynomial-Time Algorithm for NOOS}
    \begin{thm}
    \label{thm:noos}
        The NOOS is in $\P$.
    \end{thm}
    \begin{proof}

    Since $\COG(\Orp)$ is non-branching, and we consider two cases depending on whether it is a path or a cycle.

    If $\COG(\Orp)$ is a path, then every vertex in it is an articulation vertex in both $\COG(\Orp)$ and $\OG(\Orp)$. Our algorithm will process this path in a divide-and-conquer manner. Namely, for a path of length greater than 2, we pick a vertex $v$ closest to the path middle. If $v$ is oriented, we proceed as in Theorem~\ref{thm:atr_vert}. If $v$ is unoriented, we fix each of the two possible orientations, proceed as in Theorem~\ref{thm:atr_vert}, and pick the better solution among them. 
    
    A path of length at most 2 can be solved in $\bigO{|\Orp|}$ time by brute-forcing all possible orientations of the scaffolds in the path and counting how many assembly points in $\Orp$ get consistent orientations.

    The running time $T(l)$ for recursive part of the algorithm satisfies the formula: 
    $$T(l) = 
    \begin{cases}
        4\cdot T\left(\frac{l}{2}\right) + \bigO{1}, & \text{if}\ |\Orp| > 2;\\
        \bigO{|\Orp|}, & \text{if}\ |\Orp| \leq 2.
    \end{cases}$$
    From the Master theorem~\cite{Bentley1980}, we conclude that the total running time for the proposed recursive algorithm is $\bigO{|\Orp|^2}$ (or $\bigO{|\Scafs(\Ass)|^2}$ since $\COG(\Orp)$ is a path).

    If $\COG(\Orp)$ is a cycle, we can reduce the corresponding NOOS instance to the case of a path as follows. 
    First, we pick a random vertex $w$ in $\COG(\Orp)$ and replace it with new vertices $w_1$ and $w_2$ such that the edges $\{u, w\}$, $\{w, v\}$ in $\COG(\Orp)$ are replaced with $\{u, w_1\}$, $\{w_2, v\}$. 
    Then we solve the NOOS for the resulting path one or two times (depending on whether $w\in\Scafs_o(\Ass)$): once for each of possible orientations of scaffold $w$ (inherited by $w_1$ and $w_2$), 
    and then select the orientation for $w$ that produces a better result.
    \end{proof}
    
    A pseudocode for the algorithm described in the proof of Theorem~\ref{thm:noos} is given in Algorithm~\ref{algo:noos} in the Appendix.
    
\subsubsection*{An exact algorithm for the OOS}
Below we show how to solve OOS instance $(\Ass, \Orp)$ in general case, i.e., when $\COG(\Orp)$ is neither a path or a cycle. 

First we assume that there are no articulation vertices in the $\OG(\Orp)$, while the case when articulation vertices are present is addressed in the next section. Let $\BV(\Orp)$ be the set of unoriented \emph{branching vertices} (i.e., unoriented vertices of degree greater than 2) in $\COG(\Orp)$. We define a non-branching path as a path for which the endpoints are in $\BV(\Orp)$, and all internal vertices have degree 2 (e.g., $\{s_{18}, s_{23}, s_{24}, s_{25}\}$ is a non-branching path in Fig.~\ref{fig:cog_spy_cd}a). Similarly, we define a non-branching cycle as a cycle in which all vertices have degree 2, except for one vertex (called \emph{endpoint}) that belongs to $\BV(\Orp)$ and thus has degree greater than 2 (e.g., $\{s_7, s_4, s_3, s_1, s_2, s_6, s_5, s_7\}$ is a non-branching cycle in Fig.\ref{fig:cog_spy_cd}a). 

Each OOS instance induced by a non-branching path and a non-branching cycle in $\COG(\Orp)$ represents an NOOS instance, and thus can be solved in polynomial time. We iterate over all possible orientations for the endpoints of the underlying paths/cycles in the corresponding NOOS instances and solve them. A solution to the OOS instance is obtained by iterating over all possible orientations of the scaffolds represented by branching vertices in $\COG(\Orp)$ (i.e., $\BV(\Orp)$) and merging the solutions to the corresponding NOOS instances, and picking the best result. Then, the following lemma trivially holds: 

\begin{lem}
\label{lem:general_complexity}
\sloppy
The running time for the proposed algorithm is bounded by 
$
\bigO{2^{|\BV(\Orp)|}\cdot |\Scafs(\Ass)|^2}.
$
\end{lem}

\subsubsection*{An $\FPT$ algorithm for the OOS}
Thanks to Theorem~\ref{thm:atr_vert}, we can partition a given OOS instance $(\Ass, \Orp)$ into subinstances using the oriented articulation vertices. By Theorem~\ref{thm:noos}, we also know how to efficiently orient scaffolds that correspond to unoriented articulation vertices of degree 2. In this section, we address the remaining type of articulation vertices, namely unoriented articulation vertices of degree at least 3.  

Let $\AV(\Orp)\subseteq\BV(\Orp)$ be the set of unoriented articulation vertices of degree at least 3. A straightforward solution to this problem is to iterate over all possible $2^{|AV(\Orp)|}$ orientations of the scaffolds in $\AV(\Orp)$, and then use Theorem~\ref{thm:atr_vert} to partition the OOS instance $(\Ass, \Orp)$ into subinstances. Each such subinstance, in turn, can be solved using Theorem~\ref{thm:noos} or Theorem~\ref{lem:general_complexity}. Below we show how one can orient the scaffolds in $\AV(\Orp)$ more efficiently based on the dependencies between the connected subgraphs flanked by the corresponding vertices.

The set $\AV(\Orp)$ defines a set $\Con(\Orp)$ of connected subgraphs (\emph{components}) of $\COG(\Orp)$ by breaking it at the vertices from $\AV(\Orp)$, introducing copies of each articulation vertex in the resulting components (Fig.~\ref{fig:cog_spy_cd}a). We distinguish between two types of components in $\Con(\Orp)$:
\begin{itemize}
\item \emph{path bridges} forming the set $\NBPB(\Orp)\subseteq\Con(\Orp)$, i.e., components that do not contain cycles (e.g., $pb_1$ in Fig.~\ref{fig:sag_og_cog}a);
\item \emph{complex components} forming the set $\CC(\Orp)\subseteq\Con(\Orp)$, i.e., components that contain at least one cycle (e.g., $cc_2$ in Fig.~\ref{fig:cog_spy_cd}a).
\end{itemize}
Trivially we have $\CC(\Orp)\cup\NBPB(\Orp)=\Con(\Orp)$. We denote by $V(c)$ the set of vertices in a component $c \in \Con(\Orp)$. Now, we show how to solve the OOS instances induced by elements of $\Con(\Orp)$:

\paragraph*{Case $c\in\NBPB(\Orp)$} The OOS instance induced by $c$ can be solved as follows. We iterate over all possible orientations of the unoriented articulation vertices in $c$ (i.e., we need solve the OOS instance induced by $c$ at most 4 times). For each fixed orientation, since $c$ is non-branching, the OOS instance induced by $c$ is an instance of NOOS and can be solved as in Theorem~\ref{thm:noos}.

\paragraph*{Case $c\in\CC(\Orp)$} The OOS instance induced by $c$ can be solved as follows. We iterate over all possible orientations of the vertices in $\AV(\Orp)\cap V(c)$. For each fixed orientation, a solution to the OOS instance induced by $c$ can be obtained as in Theorem~\ref{lem:general_complexity} by iterating over all possible orientations of the scaffolds represented by the unoriented branching vertices in $c$ (i.e., $(\BV(\Orp) \setminus \AV(\Orp)) \cap V(c)$). 

    
\begin{figure}[t!]
\includegraphics[width=.9\textwidth]{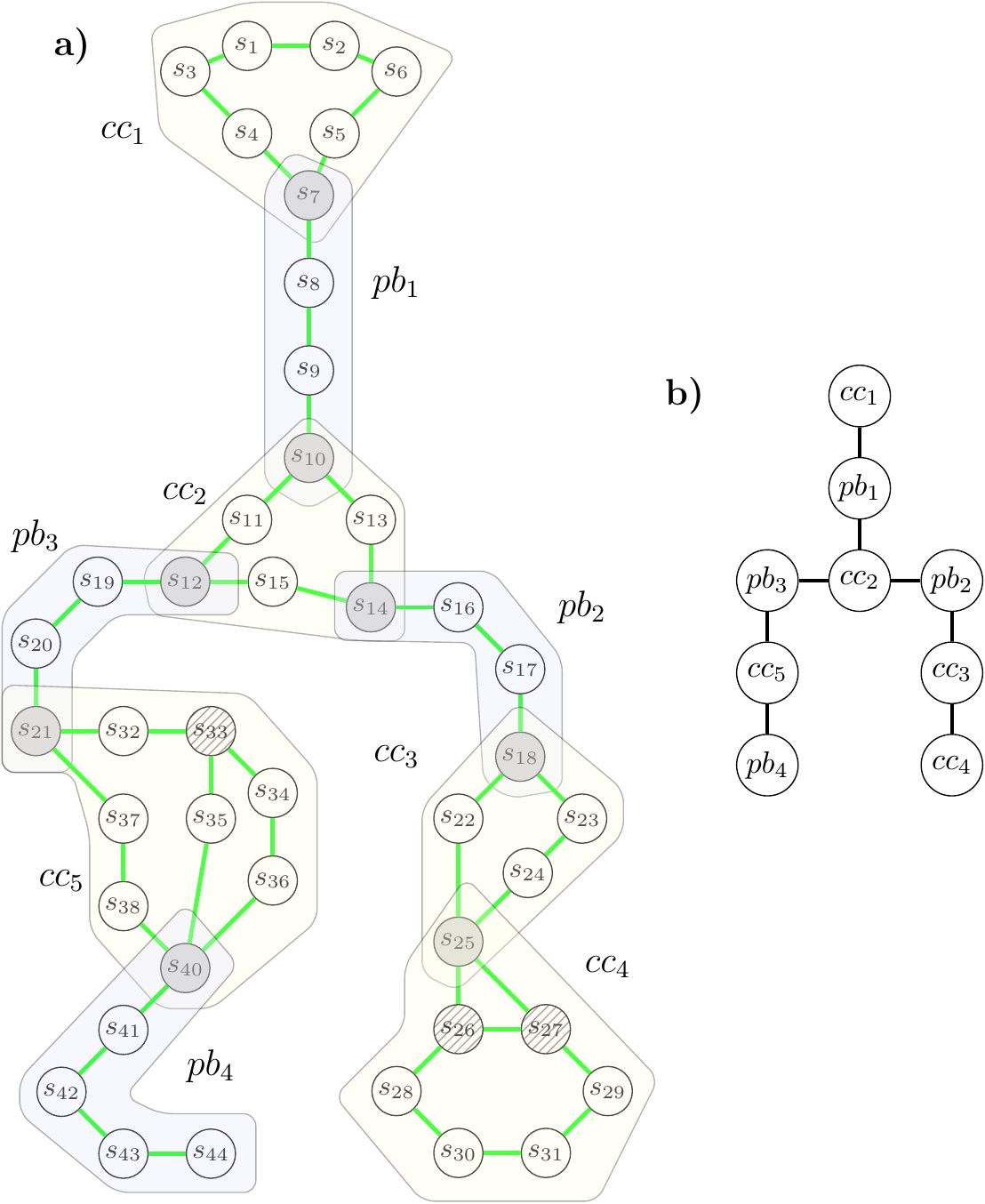}
\caption{\textbf{a)} Contracted ordered graph $\COG(\Orp)$ of a set of assembly points $\Orp$. Branching articulation vertices $\AV(\Orp) = \{s_7, s_{10}, s_{12}, s_{14}, s_{21}, s_{40}\}$ are shown as filled with gray. Branching vertices that are not articulation vertices $\BV(\Orp)\setminus\AV(\Orp) = \{s_{33}, s_{26}, s_{27}\}$ are shown as filled with line pattern. Yellow areas highlight elements of $\CC(\Orp) = \{cc_1, cc_2, cc_3, cc_4, cc_5\}$. Blue areas highlights elements of $\NBPB(\Orp) = \{pb_1, pb_2, pb_3, pb_4\}$. \textbf{b)} The subproblem tree $\ST(\Orp)$. 
}
\label{fig:cog_spy_cd}     
\end{figure}

Now, we outline how we iterate over the orientations of scaffolds in $\AV(\Orp)$. Our algorithm constructs a \emph{subproblem tree} $\ST(\Orp) = (V, E)$ (Fig.~\ref{fig:cog_spy_cd}b), where $V=\Con(\Orp)$ is the set of vertices corresponding to the set of components induced by $\AV(\Orp)$, and $E$ is the set of edges constructed iteratively. We start with $E=\emptyset$ and populate $E$ as follows: for each vertex $v\in V$ and all vertices $u\in V$, add an edge $\{ v, u \}$ if the following two conditions hold:
\begin{enumerate}
\item $v$ and $u$ share an articulation vertex in $\COG(\Orp)$ (e.g., $cc_2$ and $pb_1$ in Fig.~\ref{fig:sag_og_cog}a); and
\item $u$ is not an endpoint of any edge in $E$.
\end{enumerate}

A subproblem tree $\ST(\Orp)$ allows us to solve the original OOS instance in the bottom-up fashion. Indeed, the OOS instance corresponding to any disjoint subtrees of $\ST(\Orp)$ can be solved independently. We start with solving OOS instances that correspond to the leaves, producing solutions corresponding to different orientations of the scaffolds corresponding to articulation vertices. When the OOS instances for all children of an internal vertex $c$ in $\ST(\Orp)$ are solved, we iterate over the orientations for the scaffolds that correspond to articulation vertices in $c$ (i.e., $\AV(\Orp)\cap V(c)$) and merge the OOS solutions for $c$ with the corresponding solutions for its children. Eventually, we obtain the OOS solution for the root of $\ST(\Orp)$ and thus for the original OOS problem.

    

The following theorem states the running time of the proposed algorithm.

\begin{thm}
\label{thm:complexity}
The running time for the proposed algorithm for solving OOS instance $(\Ass,\Orp)$ is bounded by 
\begin{equation}\label{eq:complexity}
\bigO{2^{\alpha} \cdot |\Scafs(\Ass)|^2 \cdot |\CC(\Orp)|},
\end{equation}
where $\alpha=\max_{c\in\CC(\Orp)} |\BV(\Orp)\cap V(c)|$.
\end{thm}
\begin{proof}
The construction time of $\SAG(\Ass)$, $\AV(\Orp)$, $\BV(\Orp)$, $\OG(\Orp)$, $\COG(\Orp)$, $\ST(\Orp)$, and $\Con(\Orp)$ is bounded by $\bigO{|\Scafs(\Ass)|^2}$. 

The OOS instances induced by each non-branching path or cycle in $\COG(\Orp)$ are solved at most $4$ times for different orientations of the endpoints. By Theorem~\ref{thm:noos}, the total running time for processing all non-branching paths/cycles in $\COG(\Orp)$ is bounded by $\bigO{|\Scafs(\Ass)|^2}$.

By Lemma~\ref{lem:general_complexity}, each OOS instance induced by a complex component $c\in\CC(\Orp)$ can be solved in $\bigO{2^m\cdot |\Scafs(\Ass)|^2}$ time, where $m=|(\BV(\Orp) \setminus \AV(\Orp))\cap V(c)|$. 
The running time of the bottom-up algorithm is bounded by $|\Con(\Orp)|$ (i.e., the number of vertices in $\ST(\Orp)$) times the running time of the merging procedure bounded by $\bigO{2^{|\AV(\Orp)\cap V(c)|}\cdot \deg(c)}$, where $\deg(c)$ is the degree of $c$ in $\ST(\Orp)$.

Thus, the proposed algorithm can be bounded by $\bigO{2^{\alpha} \cdot |\Scafs(\Ass)|^2 \cdot |\CC(\Orp)|}$, where $\alpha=\max_{c\in\CC(\Orp)} |\BV(\Orp)\cap V(c)|$. 
\end{proof}

The proposed algorithm is an $\FPT$ algorithm. Indeed, instead of finding the best orientation by iterating over all possible orientations of the scaffolds in $\Scafs_u(\Ass)$, we iterate over all possible orientations of the scaffolds that correspond to branching vertices in $\COG(\Orp)$. Furthermore, we reduced running time of an $\FPT$ algorithm by partitioning the problem into connected components and solving them independently. 

The exponential term in \eqref{eq:complexity} accounts for the number of articulation vertices in the complex components of $\COG(\Orp)$. For real data, the exponent can become large only if many scaffolds have relative orientation with respect to three or more other scaffolds, which we expect to be a rare situation, especially when the scaffolds are long (e.g., produced by scaffolders combining paired-end and long-read data, a popular approach for the genome assembly).
    
\section{Conclusions}
In the present study, we posed the orientation of ordered scaffolds (OOS) problem as an optimization problem based on given weighted orientations of scaffolds and their pairs. We further addressed it within the earlier introduced \CAMSA framework~\cite{Aganezov2017}, taking advantage of the simple yet powerful concept of assembly points describing (semi-/un-) oriented adjacencies between scaffolds. This approach allows one to uniformly represent both orders of oriented and/or unoriented scaffolds and orientation-imposing data.

We proved that the OOS problem is $\NP$-hard when the given scaffold order represents a linear or circular genome. We also described a polynomial-time algorithm for the special case of non-branching OOS (NOOS), where the orientation of each scaffold is imposed relatively to at most two other scaffolds. Our algorithm for the NOOS problem and Theorems~\ref{thm:splitOGA}, \ref{thm:splitOGO}, and \ref{thm:atr_vert} further enabled us to develop an $\FPT$ algorithm for the general OOS problem. The proposed algorithms are implemented in the \CAMSA software version 2.

\section*{Funding}
The work of SA is supported by the National Science Foundation under grant DBI-1350041, National Institute of Health under grant R01-HG006677, and Bill and Melinda Gates Foundation. 
The work of PA and YR are partially supported by the National Science Foundation under grant DMS-1406984.
The work of NA is supported by the Government of the Russian Federation under the ITMO Fellowship and Professorship Program. 

\section*{Acknowledgments}
Partial preliminary results of the present work appeared in the proceedings of the $15^\text{th}$ Annual RECOMB Satellite Workshop on Comparative Genomics~\cite{aganezov2017a}. 

\bibliographystyle{acm}
\bibliography{library.bib}

\clearpage
\newpage

\section*{Appendix. Pseudocodes}
In the algorithms below we do not explicitly describe the function \textsc{OrConsCount}, which takes 4 arguments: 
\begin{enumerate}
    \item a subgraph $c$ from $\COG(\Orp)$ with 1 or 2 vertices;
    \item a hash table $so$ with scaffolds as keys and their orientations as values;
    \item a set of orientation imposing assembly points $\Orp$;
    \item an assembly $\Ass$ 
\end{enumerate}
and counts the assembly points from $\Orp$ that have consistent orientation with $\Ass$ in the case where scaffold(s) corresponding to vertices from $c$ were to have orientation from $so$ in $\Ass$. With simple hash-table based preprocessing of $\Ass$ and $\Orp$  this function runs in $\bigO{n}$ time, where $n$ is a number of assembly points in $\Orp$ involving scaffolds that correspond to vertices in $c$. So, total running time for all invocations of this function will be $\bigO{|\Orp|}$ (i.e., $\bigO{|\Scafs(\Ass)|^2}$).

\begin{algorithm}
\caption{Solving the NOOS for complex component}
\label{algo:noos}
\begin{algorithmic}[1]
    \Function{SolveNOOSforCC}{$c$, $\Orp$, $so$, $\Ass$}
    \State $score \gets\ 0$
    \If{$c$ has less than 3 vertices}
    \State \textbf{return} \textsc{OrConsCount}($c$, $so$, $\Orp$, $\Ass$)
    \EndIf
    \State $p_1, p_2\gets$ split $c$ into two paths of equal length, at vertex $s$
    \State $var\gets$ empty hash table
    \For{$or$ in $\{\rightarrow, \leftarrow\}$}
            \State $so[s]\gets or$
            \State $var[or]$ = \textsc{SolveNOOSforCC}($p_1$, $\Orp$, $so$)
            \State $var[or]$ += \textsc{SolveNOOSforCC}($p_2$, $\Orp$, $so$)
            \State $var[or]$ += \textsc{OrConsCount}($s$, $so$, $\Orp$, $\Ass$)
    \EndFor
    \State $score\gets$ maximum value in $var$
    \State $or\gets$ key corresponding to the maximum value in $var$
    \State $so[s]\gets or$
    \State \textbf{return} $score$
    \EndFunction
\end{algorithmic}
\end{algorithm}
    
\begin{algorithm}
\caption{Solving the NOOS for instance $(\Ass, \Orp)$ }
\begin{algorithmic}[1]
    \State
    \Function{SolveNOOS}{$\Ass$, $\Orp$}
    \State $so\gets$ hash table with $\Scafs(\Ass)$ as keys and their orientation as values
    \State $cog\gets$ $\COG(\Orp_o)$
    \State $var\gets$ empty hash table
    \For{each connected component $c$ in $cog$}
        \If{$c$ is a cycle}
            \State pick a random vertex $w$ in $c$
            \State remove two edges $\{u, w\}$, $\{w, v\}$ from $c$
            \State add two edges $\{u, w_1\}$, $\{w_2, u\}$ to $c$
            \For{$or$ in $\{\rightarrow, \leftarrow\}$}
                \State $so[w_1], so[w_2]\gets or, or$
                \State $var[(w, or)]\gets$ \textsc{SolveNOOSforCC}($c$, $\Orp$, $so$, $\Ass$)
            \EndFor
        \Else
            \State $s_1, s_2\gets$ scaffolds, corresponding extremities of $c$
                \For{$or_1$ in $\{\rightarrow, \leftarrow\}$}
                \For{$or_2$ in $\{\rightarrow, \leftarrow\}$}
                    \State $so[s_1], so[s_2]\gets or_1, or_2$
                    \State $r\gets$ \textsc{SolveNOOSforCC}($c$, $\Orp$, $so$, $\Ass$)
                    \State $var[(s_1, or_1)], var[(s_2, or_2)]\gets r, r$ 
                \EndFor
            \EndFor
        \EndIf
    \EndFor
    \State \textbf{return} $var$
    \EndFunction
\end{algorithmic}

\end{algorithm}

\begin{algorithm}
\caption{Checking if a given assembly has a non-conflicting realization}
\label{algo:checkNCR}
\begin{algorithmic}[1]
\Procedure{HasNonConflictingRealization}{$\Ass$}
\State $t\gets$ hash table with $\Scafs(\Ass)$ as keys and empty linked lists as values
\For{$p$ in $\Ass$}
    \State append $p$ to the list at $t$[$\sn(p, 1)$]
    \State append $p$ to the list at $t$[$\sn(p, 2)$]
    \If{length of list at $t$[$\sn(p, 1)$] $> 2$ \textbf{or} length of list at $t$[$\sn(p, 2)$] $> 2$}
        \State \textbf{return} False
    \EndIf
\EndFor
\For{$s$ in $\Scafs(\Ass)$}
    \If{length of list at $t$[$s$] = 2 \textbf{and} assembly points in $t$[$s$] are conflicting }
    \State \textbf{return} False
    \EndIf
\EndFor
\State \textbf{return} True
\EndProcedure
\end{algorithmic}
\end{algorithm}

\begin{algorithm}
\caption{Computing $\Scafs_u(\Ass)$}
\label{algo:computingSu}
\begin{algorithmic}[1]
\Procedure{UnorientedScaffolds}{$\Ass$}
\State $S\gets$ empty set
\If{not HasNonConflictingRealization($\Ass$)}
    \State \textbf{return} $S$
\EndIf
\State $t\gets$ hash table with $\Scafs(\Ass)$ as keys and empty linked lists as values
\For{$p$ in $\Ass$}
    \State append $p$ to the list at $t$[$\sn(p, 1)$]
    \State append $p$ to the list at $t$[$\sn(p, 2)$]
\EndFor
\For{$s$ in $\Scafs(\Ass)$}
    \State $flag\gets$ True
    \For{$p$ in the list at $t$[$s$]}
        \If{$p$ involves $s$ with orientation}    
            \State $flag\gets$ False
        \EndIf
    \EndFor
    \If{$flag$}
        \State add $s$ to $S$
    \EndIf
\EndFor
\State \textbf{return} $S$
\EndProcedure
\end{algorithmic}
\end{algorithm}

\end{document}